\documentclass[letterpaper]{article} 
\usepackage{aaai24}  
\usepackage{times}  
\usepackage{helvet}  
\usepackage{courier}  
\usepackage[hyphens]{url}  
\usepackage{graphicx} 
\usepackage{natbib}  
\usepackage{caption} 
\usepackage[algo2e,linesnumbered,ruled]{algorithm2e}
\usepackage{algorithm}
\usepackage{algorithmic}
\usepackage{amsmath}
\usepackage{amsthm}
\usepackage{amssymb}
\usepackage{mathtools}
\usepackage{color}
\usepackage{multicol}

\newtheorem{theorem}{Theorem}
\newtheorem{lemma}[theorem]{Lemma}

\newtheorem{proposition}[theorem]{Proposition}

\newtheorem{definition}[theorem]{Definition}

\DeclareMathOperator{\CPT}{CPT}
\DeclareMathOperator{\Inst}{Inst}

\usepackage{newfloat}
\usepackage{listings}
\DeclareCaptionStyle{ruled}{labelfont=normalfont,labelsep=colon,strut=off} 
\floatstyle{ruled}
\newfloat{listing}{tb}{lst}{}
\floatname{listing}{Listing}
\title{Approximation Algorithms for Preference Aggregation Using CP-Nets}
\author {
    Abu Mohammad Hammad Ali,
    Boting Yang,
    Sandra Zilles
}
\affiliations {
    University of Regina\\
    hammad@uregina.ca, boting.yang@uregina.ca, sandra.zilles@uregina.ca
}

\usepackage{bibentry}

\begin{document}

\maketitle

\begin{abstract}
This paper studies the design and analysis of approximation algorithms for aggregating preferences over combinatorial domains, represented using Conditional Preference Networks (CP-nets). Its focus is on aggregating preferences over so-called \emph{swaps}, for which optimal solutions in general are already known to be of exponential size. We first analyze a trivial 2-approximation algorithm that simply outputs the best of the given input preferences, and establish a structural condition under which the approximation ratio of this algorithm is improved to $4/3$. We then propose a polynomial-time approximation algorithm whose outputs are provably no worse than those of the trivial algorithm, but often substantially better. A family of problem instances is presented for which our improved algorithm produces optimal solutions, while, for any $\varepsilon$, the trivial algorithm can\emph{not}\/ attain a $(2-\varepsilon)$-approximation. These results may lead to the first polynomial-time approximation algorithm that solves the CP-net aggregation problem for swaps with an approximation ratio substantially better than $2$.
\end{abstract}

\section{Introduction}

The goal of preference aggregation is to find, given a set of individual rankings over objects called \emph{outcomes}, either the best collective outcome or the best collective ranking over the outcomes. Preference aggregation has applications in the domain of recommender systems, multi-criteria object selection, and meta-search engines \citep{dwork2001rank}. In this paper, we study preference aggregation over combinatorial domains, using so-called Conditional Preference Networks (CP-nets, \citep{boutilier2004cp}) as a compact representation model for outcome rankings. A combinatorial domain defines outcomes as vectors of attribute-value pairs. By expressing conditional dependencies between attributes, a CP-net represents the preferences over a large number of outcome pairs using compact statements. For example, given five attributes $V_1$ through $V_5$, such a statement might be \emph{``Given value 1 in attribute $V_4$, I prefer value 0 over value 1 in attribute $V_5$.''} This statement means that all outcomes with value assignment $(1,0)$ for attributes $(V_4,V_5)$ are preferred over those with value assignment $(1,1)$, irrespective of their values in the attributes $V_1$ through $V_3$. This saves the resources needed for explicitly listing preferences between pairs of outcomes with various values in $V_1$ through $V_3$. In this example, $V_4$ is called a parent of $V_5$; in general, an attribute can have more than one parent.

In this paper, we study the problem of aggregating multiple CP-nets $N_1$, \dots, $N_t$ (without cyclic dependencies between attributes) into a single CP-net $N$ which forms the best possible consensus of $N_1$, \dots, $N_t$ in terms of the associated outcome rankings. One objective might be to minimize the total number of outcome pairs that are ordered differently by $N$ and $N_s$, summed up over all $s\in\{1,\ldots,t\}$. However, feasible solutions to this problem are unlikely to be found, partly because even determining how a given CP-net orders a given outcome pair is NP-hard \cite{boutilier2004cp}. For this and other reasons, \citet{ali2021aggregating} propose to focus on an objective function that counts only \emph{swaps}\/ that are ordered differently by $N$ and $N_s$. A swap is a pair of outcomes that differ only in the value of a single attribute; deciding how a CP-net orders any given swap can be done in polynomial time \cite{boutilier2004cp}.

We adopt the objective function proposed by \citet{ali2021aggregating}.\footnote{Technically, we limit this objective function to only swaps in which the two outcomes differ in a fixed attribute $V_n$, since it is sufficient to reduce the CP-net aggregation problem to an attribute-wise aggregation problem. Technical details will follow.} However, they showed that this function cannot be optimized in polynomial time; in particular, sometimes the size of the only optimal solution is exponential in the size of the input. In a preliminary result, we show that only the number of input CP-nets (not the number of attributes) contribute to the hardness of optimal aggregation. Motivated by these results, we study efficient approximation algorithms for Ali et al.'s objective function. 

A first (trivial) approximation algorithm simply outputs an input that obtains the smallest value of the objective function among all inputs. It is a well-known fact that this trivial algorithm guarantees a 2-approximation, but \citet{endriss2014binary} showed that the bound of 2 cannot be improved. Our first main result states that the trivial algorithm obtains an approximation ratio of $4/3$ in case the inputs satisfy a natural (yet limiting) symmetry condition.

We then propose an improved algorithm that, given an attribute $V_n$, considers the parent sets for $V_n$ in the input CP-nets $N_1$, \dots $N_t$. For each such parent set $P$, the algorithm first computes a provably optimal aggregate among all CP-nets that use $P$ as a parent set for $V_n$. It then computes the objective value for each resulting aggregate and outputs one with the smallest such value. Our formal results on this improved algorithm entail that it is guaranteed to be no worse than the trivial algorithm. We then define a family of problem instances for which the improved algorithm is optimal (i.e., has approximation ratio 1), while the trivial algorithm has ratio at least $3/2$. In particular, the ratio of the trivial algorithm cannot be bounded below 2 for this family.
Whether the improved algorithm obtains a ratio of at most 4/3 in general, remains an open question.

We hope to thus initiate a line of research on approximation algorithms for CP-net aggregation, and to enrich the research on approximation algorithms in the more general context of binary aggregation \cite{endriss2014binary}.

\section{Related Work}\label{sec:related}

A substantial body of research addresses preference aggregation using explicit preference orders, represented as permutations over the outcome space, with the goal of finding a permutation that minimizes some objective function \cite{sculley2007rank,dwork2001rank,dinu2006efficient,bachmaier2015hardness}. For outcomes defined over a combinatorial domain, one works with compact preference representations \cite{airiau2011aggregating}, e.g., CP-nets \cite{boutilier2004cp}, LP-Trees \cite{booth2010learning}, utility-based models, or logical representation languages \cite{lang2004logical}. 

One approach to preference aggregation using CP-nets is that of mCP-nets \cite{rossi2004mcp, lukasiewicz2016complexity, lukasiewicz2019complexity, lukasiewicz2022complexity}. In this approach, the input is a set of partial CP-nets. No single aggregate model is constructed. Instead, preference reasoning tasks such as outcome ordering and optimization are performed using some voting rule on the set of input CP-nets, and all input CP-nets must be stored.

Given a set of acyclic input CP-nets, \citet{lang2007vote} proposes to elicit votes sequentially over the attributes, using the value assigned to any parent attributes, 
thus constructing a consensus outcome. 
\citet{xia2007sequential,xia2007strongly} showed that sequential voting may lead to paradoxical outcomes, which can be avoided by assuming some linear ordering over the attributes. 
Further work on sequential voting was presented by \citet{lang2009sequential,grandi2014aggregating}. A similar approach considers voting over general CP-nets using a hypercube-wise decomposition, \cite{xia2008voting,conitzer2011hypercubewise,li2011majority,brandt2016handbook}. 
Lastly, \cite{cornelio2013updates,cornelio2015reasoning,cornelio2021reasoning} address aggregating CP-nets using PCP-nets, which are an extension to CP-nets that allow for probabilistic uncertainty. 

In the present study, we are not interested in finding the joint best outcome of the given input CP-nets. Instead we want to create a consensus preference ordering (over all so-called swap pairs, i.e., pairs of outcomes that differ only in a single attribute) that best aggregates the given preference orders, under the constraint that this consensus ordering can be represented as a CP-net. Our approach is similar to that of \citet{ali2021aggregating}, in that we treat preference aggregation as an optimization problem where the input profile and the optimal output are both represented using CP-nets. In contrast to the mCP-nets or PCP-nets approach, this avoids storing all input CP-nets and allows for applying existing CP-net algorithms for reasoning about preferences. However, \citet{ali2021aggregating} showed that there is no polynomial-time algorithm solving the problem that we focus on. This motivates us to study approximation algorithms for said problem.

Hardness results for aggregation were also established outside the context of CP-nets. For rank aggregation over explicit total or partial orders over the outcome space, \citet{dwork2001rank} showed that optimizing based on the cumulative pairwise distance from each input ordering, known as Kemeny optimization, is intractable. In the related field of judgement aggregation, \citet{endriss2012complexity} proved that distance-based aggregation is intractable, which motivates a simple 2-approximation algorithm \cite{endriss2014binary}. \citet{ailon2008aggregating,ailon2010aggregation} studied (expected) approximation ratios  of a randomized algorithm as well as of a linear programming approach.

\section{Preliminaries}


\citet{boutilier2004cp} define a CP-net $N$ as a directed graph, in which the vertex set $\mathcal{V}=\{V_1,\ldots,V_n\}$ is a set of $n$ binary attributes, with $\{0,1\}$ as the set of possible values of each attribute $V_{i}$. A preference over $V_{i}$ is now simply one of the two possible total orders over $\{0,1\}$. An edge $(V_{j},V_{i})$ means that the user's preference over $V_{i}$ depends on the value of $V_{j}$, in which case $V_{j}$ is called a parent of $V_{i}$. We focus solely on \emph{acyclic}\/ CP-nets. By $Pa(N,V_{i})$ one denotes the set of parents of $V_{i}$ in a CP-net $N$. If $Pa(N,V_i)=\emptyset$ for all $V_i\in\mathcal{V}$, then $N$ is called \emph{separable}. 

For each $V_i$, the user's conditional preferences over $\{0,1\}$ are listed in a Conditional Preference Table (CPT), denoted $\CPT(N,V_{i})$. For example, suppose $V_i$ has only one parent, namely $V_j$. Then the CPT entry $0:0\succ 1$ is read  ``If $V_j$ has the value $0$, then $0$ is preferred over $1$ for $V_i$.'' Since a CPT for $V_i$ specifies at most one preference per assignment of values to $Pa(N, V_{i})$, it lists at most $2^k$ preferences, called \emph{CPT rules}, where $k=|Pa(N,V_i)|$.  The \emph{size}\/ of a CPT is the total number of its rules. An incomplete CPT is one that is of size strictly less than $2^k$. In this paper, we always assume implicitly that CP-nets are \emph{complete}, i.e., any CPT contains the maximum possible number of rules.\footnote{This would allow us to represent CPTs more compactly, by only listing those rules whose preference over $\{0,1\}$ is less frequent. However, all formal results in this paper hold irrespective of whether one represents CPTs this way or by listing all rules.} This assumption can be limiting for algorithmic studies, but is not uncommon in the literature (see, e.g., \citep{AlanaziMZ20,ali2021aggregating}). Note that our study can be generalized to incomplete CP-nets with some additional effort, yet without major conceptual differences.

An instantiation of a set $\mathcal{V}'\subseteq \mathcal{V}$ is an assignment of values to each attribute in $\mathcal{V}'$; then $\Inst(\mathcal{V}')$ denotes the set of all instantiations of $\mathcal{V}'$. Note that $\Inst(\emptyset)$ contains only the empty tuple. Assuming a fixed order over $\mathcal{V}$, each element $\gamma\in\Inst(\mathcal{V}')$ is simply a boolean vector with $|\mathcal{V}'|$ components, where $\gamma[V_i]$ denotes the value of $\gamma$ in $V_i$, if $V_i\in \mathcal{V}'$. Given $\mathcal{V}',\mathcal{V}''\subseteq\mathcal{V}$ and $\gamma'\in\Inst(\mathcal{V}')$, and $\gamma''\in\Inst(\mathcal{V}'')$, we say $\gamma'$ is \emph{consistent}\/ with $\gamma''$ (and vice versa) iff $\gamma'[V]=\gamma''[V]$ for all $V\in\mathcal{V}'\cap\mathcal{V}''$. Elements of $\Inst(\mathcal{V})$ are called outcomes. Thus, any outcome $o$ corresponds to the vector $(o[V_1],\ldots,o[V_n])$. An outcome pair $(o,o')$ is called a \emph{swap}\/ over $V_i$ if $o,o'$ differ only in their value in $V_i$, and $o[V_i]=0$, $o'[V_i]=1$.

In general, consider a CPT rule for $V_i$ of the form $\gamma: b\succ b'$, where $\gamma\in\Inst(Pa(N,V_i))$ and $\{b,b'\}=\{0,1\}$. Here, $\gamma$ is called the \emph{context}\/ of the CPT rule. The rule is interpreted using the \textit{ceteris paribus} assumption: if $(o,o')$ or $(o',o)$ is a swap over $V_i$, $o[V_i]=b$, $o'[V_i]=b'$, and $o[V_j]=o'[V_j]=\gamma[V_j]$ for all $V_j\in Pa(N,V_i)$, then $o$ is preferred over $o'$, written $o\succ o'$. This way, a complete CP-net orders all swap pairs, i.e., for each swap $(o,o')$ over any attribute $V_i$, the CP-net entails either $o\succ o'$ or $o'\succ o$. By identifying $\succ$ with its transitive closure, one obtains a partial preference order over the space $\Inst(\mathcal{V})$ of all outcomes.

\subsubsection{Problem Formulation}

The focus of this paper lies on forming an aggregate CP-net, given a tuple $T=(N_1,\ldots, N_t)$ of input CP-nets. This aggregate is supposed to represent a consensus among the preferences of the underlying individual input CP-nets. This raises the question of how to assess how well a CP-net $N$ represents a consensus between several input CP-nets. One measure could be the total number of triples $(s,o,o')$ where $1\le s\le t$ and $(o,o')$ is any outcome pair, such that $N$ orders $(o,o')$ differently than $N_s$. However, \citet{boutilier2004cp} showed that deciding whether a CP-net entails $o\succ o'$ is NP-hard in general, which substantially hinders the design of efficient algorithms for constructing consensus CP-nets that minimize such measure. Moreover, not every outcome pair is ordered by every CP-net, which makes it non-trivial to even define when ``$N$ orders $(o,o')$ differently than $N_s$'' \citep{ali2021aggregating}.

By contrast, deciding whether a CP-net entails $o\succ o'$, for arbitrary swaps $(o,o')$ can be done in polynomial time \citep{boutilier2004cp}. Also, since every complete CP-net orders every swap, it is easy to test whether two CP-nets order a given swap in two different ways. We thus follow the approach by \cite{ali2021aggregating}, namely to aggregate CP-nets with respect to a measure that counts the number of cases in which a proposed consensus CP-net disagrees with a given input CP-net. Specifically, given two CP-nets $N$ and $N_s$, we define the swap disagreement $\Delta(N,N_s)$ as the number of swaps $(o,o')$ such that $N$ and $N_s$ order $(o,o')$ differently, i.e., one of them entails $o\succ o'$, while the other entails $o'\succ o$. Our objective function for a CP-net $N$, given a tuple $T=(N_1,\ldots, N_t)$ of input CP-nets, then evaluates to
\[
f_T(N)=\sum\nolimits_{1\le s\le t}\Delta(N,N_s)\,.
\]
Our goal is the study of algorithms that, given $T$, aim at constructing an $N$ that minimizes this objective function. Note that $f_T(N)$ can be calculated in polynomial time, given $T$ and $N$ \cite{ali2021aggregating}.

One further advantage of focusing on swaps rather than general outcome pairs is that the CPT for an attribute $V_i$ alone determines how a CP-net orders a swap over $V_i$. Thus aggregating CP-nets can be done by aggregating their CPTs for each attribute separately. Therefore, we will henceforth overload the notation $N$ and $N_s$ (which so far only represented CP-nets) to refer to CPTs for the fixed attribute $V_n$, and we will focus only on aggregating CPTs for $V_n$.


In sum, this paper focuses on (efficient) algorithms that produce (optimal or approximate) solutions to the following problem, called the \emph{CPT aggregation problem}: 
\begin{itemize}
\item input: A tuple $T=(N_1,\ldots,N_t)$ of CPTs for an attribute $V_n$, over a set $\mathcal{V}=\{V_1,\ldots,V_{n-1}\}$ of $n-1$ potential parent attributes. (We call any such $T$ a \emph{problem instance}.)
\item desired output: A CPT $N$ for attribute $V_n$, over 
$\mathcal{V}$, that minimizes $f_T(N)$.
\end{itemize}

CPT aggregation is a special case of binary aggregation \citep{endriss2014binary}. In binary aggregation, one fixes a set of issues $I = \{1,\ldots,m\}$. A problem instance is a tuple of ballots, i.e., of elements of $\{0,1\}^m$, and the goal is to find the best collective ballot with a 0/1 vote for each issue. 
For CPT aggregation, each swap $(o,o')$ would be an issue, and for each issue we would have one of two possible orderings. 

\subsubsection{Matrix Representation for CP-net Aggregation}

Consider $T=(N_1,\ldots, N_t)$, $1\le s\le t$, and any swap $(o,o')$ over $V_n$. There is a context  $\gamma\in\Inst(Pa(N_s,V_n))$ such that $o[V]=o'[V]=\gamma[V]$ for each $V\in Pa(N_s,V_n)$. Given $\gamma$, $N_s$ entails $o \succ o'$ iff $\CPT(N_s,V_n)$ has the rule $\gamma: 0 \succ 1$. We thus encode the preference (called \emph{vote}\/) of any given CPT on any given swap with a boolean value: it is $0$ if $N_s$ entails $o \succ o'$, and 1 otherwise. We also encode the $2^{n-1}$ swaps over $V_n$ with their corresponding bit strings over the attributes $V_1$, \dots, $V_{n-1}$. This encodes the votes of all input CPTs on all swaps using a $2^{n-1} \times t$ boolean matrix $M(T)$. Any given row represents all the CPT votes for one swap, and any given column represents the votes of one CPT on all swaps. We use $M(T)_{\mu\nu}$ to denote the vote on swap $\mu$ by CPT $\nu$. 
Certain sub-matrices of $M(T)$ have useful interpretations. For $\mathcal{V}' \subseteq \mathcal{V} \setminus \{V_n\} $, $|\mathcal{V}'|=k$, and some context $\gamma$ of $Inst(\mathcal{V}')$, the sub-matrix $M'$ corresponding to only the swaps $(o,o')$ with $o[V]=o'[V]=\gamma[V]$ for all $V\in\mathcal{V}'$ contains the $2^{n-k-1}$ rows with all possible instantiations of $\mathcal{V} \setminus (\mathcal{V}' \cup \{V_n\})$, and all $t$ columns. For a sub-tuple $\tau$ of $T$, the sub-matrix $M'$ corresponding to only the votes of input CPTs in $\tau$ contains all $2^{n-1}$ rows, and the $|\tau|$ columns corresponding to CPTs in $\tau$. We now introduce some definitions based on this matrix representation.

\begin{definition}
\label{def:freq1}
$freq_M(1 \succ 0)$ denotes the number of votes in a matrix $M$ encoded by $1$ and $freq_M(0 \succ 1)$ denotes the number of votes in a matrix $M$ encoded by $0$. In particular, if $M$ has $2^{n-1}$ rows and $t$ columns, 
$$ freq_M(1 \succ 0) = \sum\nolimits_{0 \le \mu < 2^{n-1}} \sum\nolimits_{1 \le \nu \le t} M_{\mu\nu} $$
$$ freq_M(0 \succ 1) = t \cdot 2^{n-1} - freq_M(1 \succ 0) $$
\end{definition}

For any given swap $(o,o')$, a given row of $M(T)$ gives us the votes of the $t$ input CPTs for $(o,o')$. Hence, we call the corresponding row vector of length $t$ the \emph{voting configuration}\/ (of $T$ for swap $(o,o')$). There are $2^t$ possible bit strings of length $t$, but not all of them necessarily occur as voting configurations in $M(T)$. 

\section{Optimal Solutions}

\citet{ali2021aggregating} showed that the CPT aggregation problem cannot be solved optimally in polynomial time, simply because in some cases the \emph{size}\/ of any optimal solution is exponential in the size of the input tuple $T$ (measured in terms of the total number of CPT rules in $N_1,\ldots,N_t$):

\begin{theorem}[\citet{ali2021aggregating}]\label{thm:hardness}
  There is a family $\mathcal{F}_{bad}=(T_n)_{n\in\mathbb{N}}$ of problem instances, such that any $N^{*}$ minimizing $f_{T_n}(N^{*})$ is of size exponential in the size of $T_n$. 
\end{theorem}

The family $\mathcal{F}_{bad}$ of problem instances witnessing Theorem~\ref{thm:hardness} contains, for each $n\ge 4$, a tuple $T_n=(N^n_1,\ldots,N^n_{n-1})$ of $t=n-1$ input CPTs, where $Pa(N_s,V_n)=\{V_s\}$. Every optimal aggregate CPT for $T_n$ must have the full set $\{V_1,\ldots,V_{n-1}\}$ as a parent set and thus have $2^{n-1}$ CPT rules, while  $T_n$ contains only $2(n-1)$ rules (two rules per input CPT).

On the positive side, \citet{ali2021aggregating} noted that optimal CPT aggregation is possible in polynomial time for the subset of problem instances in which the smallest input parent set has at most $k$ attributes less than the union of all input parent sets, for some fixed $k$.
We here offer an additional positive result, based on the observation that the design of $\mathcal{F}_{bad}$ critically hinges on the number $t$ of input CPTs growing linearly with the number $n$ of attributes.

\begin{proposition}
\label{prop:fixedt}
Let $\mathcal{F}_{t\in O(1)}$ be a family of problem instances $T=(N_1,\ldots,N_t)$ for $V_n$ over $\mathcal{V}$, where $t$ is a constant. Then there is a linear-time algorithm that optimally solves the CPT aggregation problem for $\mathcal{F}_{t\in O(1)}$.
\end{proposition}
\begin{proof}
The size of the input is in $\Theta(2^{n'})$ where $n'=\max_{1\le s\le t}|Pa(N_s,V_n)|$. Since $t$ is a constant, we have $n'=\Theta(|\bigcup_{1\le s\le t}Pa(N_s,V_n)|)$. Now consider an algorithm that constructs a CPT over $P:=\bigcup_{1\le s\le t}Pa(N_s,V_n)$; for each context $\gamma$ over $P$, it checks each of the $t$ input CPTs and determines whether $1\succ 0$ is the majority preference over the input CPTs for context $\gamma$. If yes, $\gamma:1\succ 0$ is added to the output CPT; otherwise $\gamma:0\succ 1$ is added to the output CPT. Clearly, this algorithm optimally solves the CPT aggregation problem for $\mathcal{F}_{t\in O(1)}$ in linear time. 
\end{proof}

In particular, the difficulty of scaling with the number $t$ of input parent sets is the only cause for the hardness result in Theorem~\ref{thm:hardness}---scaling with the number $n$ of attributes does not pose any problems to efficient optimization. Given the overall hardness of optimally solving the CPT aggregation problem, the main focus of this paper is on efficiently constructing approximate solutions. First, we will look at obtaining approximate solutions by simply picking the best input CPT from the tuple $T$ given as problem instance.

\section{Best Input CPTs as Approximate Solutions}

As shown by \citet{endriss2012complexity}, judgement aggregation modeled as a distance-based optimization problem is intractable. This motivated the work by \citet{endriss2014binary}, which still aims at a distance optimization approach, but restricts the solution space to the inputs provided, aiming to find what they call the most representative voter. For our problem this is equivalent to using the input CPT minimizing the sum of pairwise distances from every other input CPT as an approximation to the optimal consensus CPT. The paper discusses three approaches to guide the selection of input to be used as the consensus. Two of these rules are shown to be 2-approximations of the optimal solution with a distance minimization approach. However, the paper also establishes that neither these rules, nor any other rule restricted to the input ballots submitted, can guarantee a better approximation. Their result immediately carries over to CPT aggregation.

\begin{theorem}\label{thm:2approximation}
    Let $T=(N_1,\ldots,N_t)$ be any problem instance and $N$ any optimal solution for $T$. Then $\min\{f_T(N_s)\mid 1\le s\le t\}< 2f_T(N)$. Moreover, for every $\varepsilon>0$, there exists a problem instance $T_\varepsilon=(N^\varepsilon_1,\ldots,N^\varepsilon_{t_\varepsilon})$ such that $\min\{f_{T_\varepsilon}(N^\varepsilon_s)\mid 1\le s\le t_\varepsilon\}> (2-\varepsilon)f_{T_\varepsilon}(N_\varepsilon)$, where $N_\varepsilon$ is any optimal solution for $T_\varepsilon$.
\end{theorem}
\begin{proof}
    Since every problem instance of CPT aggregation is also a problem instance of binary aggregation, the first statement follows directly from the corresponding result by \citet{endriss2014binary}. The second statement likewise follows from \citep{endriss2014binary}, since the proof therein of the corresponding binary aggregation statement uses a problem instance for which both the instance itself and the optimal solution can be cast as binary CPTs.\footnote{We will get back to this problem instance in Theorem~\ref{thm:2-epsilon}.}
\end{proof}

Since one can calculate $f_T(N_s)$ for all $s\in\{1,\ldots,t\}$ in polynomial time, Theorem~\ref{thm:2approximation} trivially yields a polynomial-time 2-approximation algorithm. By exploiting the special structure of CPTs (as opposed to general ballots in binary aggregation), we can present an improved approximation ratio for the special case of so-called \emph{symmetric}\/ CPTs.

\begin{definition}
\label{def:symmetry}
A CPT $N_s$ for $V_n$, with $|Pa(N_s,V_n)|=k$, is called symmetric iff its corresponding column vector (in matrix representation) has an equal number of zeros and ones in any sub-matrix corresponding to some fixed context of a proper subset of $Pa(N_s,V_n)$.
\end{definition}

In particular, when all input CPTs are symmetric and have pairwise disjoint parent sets, the trivial algorithm witnessing Theorem~\ref{thm:2approximation} is guaranteed to yield a $4/3$-approximation. Moreover, \emph{any}\/ input CPT yields an equally good approximation in this case:

\begin{theorem}\label{thm:symmetric}
    Let $t\ge 3$ and $T=(N_1,\ldots,N_t)$ be a problem instance in which every $N_s$ is symmetric, such that $Pa(N_s,V_n)\cap Pa(N_{s'},V_n)=\emptyset$ for $s\ne s'$. Let $N$ be any optimal solution for $T$. Then, for all $s\in\{1,\ldots,t\}$, we have $f_T(N_{s})=\min\{f_T(N_{s'})\mid 1\le s'\le t\}\le \frac{4}{3}f_T(N)$.
\end{theorem}

The proof of Theorem~\ref{thm:symmetric} is sketched via Lemmas \ref{lem:opt}--\ref{lem:even_t}, with details given in the appendix. Note that Lemmas~\ref{lem:SameCount}--\ref{lem:even_t} assume the same premises as Theorem~\ref{thm:symmetric}.

\begin{lemma}
\label{lem:opt}
Let $T'=(N'_1,\ldots,N'_{t})$ be any problem instance of $t$ CPTs. Assuming each of the $2^t$ voting configurations occurs exactly once in $M(T')$, and $N$ is an optimal solution for $T'$, we have 
\begin{equation*}
f_{T'}(N)=\begin{cases}
          2 \cdot \sum_{\kappa = 0}^{c-1} \kappa \binom{2c-1}{\kappa} &\text{if }t=2c-1 \\
          2 \cdot \sum_{\kappa = 0}^{c-1} \kappa \binom{2c}{\kappa} + c \binom{2c}{c}  &\text{if }t=2c \\
     \end{cases}
\end{equation*}
\end{lemma}

\begin{proof}
Each possible voting configuration can be represented by some bit string of length $t$. Consider a voting configuration with $\kappa$ CPTs voting $0 \succ 1$ and $t-\kappa$ CPTs voting $1 \succ 0$, $\kappa \le t-\kappa$. Clearly, $N$ makes $\kappa$ errors on this, $0 \le \kappa \le \lfloor \frac{t}{2} \rfloor $. The total error of $N$ on all voting configurations where $0 \succ 1$ is in the minority is $\sum_{\kappa=0}^{\lfloor \frac{t}{2} \rfloor} \kappa \cdot \binom{t}{\kappa} $. Accounting for those where $1 \succ 0$ is in the minority, and (in case $t$ is even) those where $0 \succ 1$ and $1 \succ 0$ are equally frequent, we obtain the desired expression for $f_{T'}(N)$.
\end{proof}

\begin{lemma}
\label{lem:SameCount}
Each of the $2^t$ voting configurations of $T$ is the row vector for exactly $2^{n-t-1}$ swaps.
\end{lemma}

\begin{proof}
Let $\mathcal{U} = (\mathcal{V} \setminus\{V_n\})\setminus \bigcup_{1 \le s \le t} Pa(N_s,V_n)$. Consider a bit string of length $t$ where the $s$-th bit represents the ordering entailed by $N_s$. Assume the $s$-th bit is $0$. Given symmetry, we know that exactly $2^{p_s-1}$ contexts of $Pa(N_s,V_n)$ entail $0 \succ 1$, where $p_s=|Pa(N_s,V_n)|$. Thus, $2^{p_s-1}$ of the voting configurations have $0$ in the $s$-th bit. Extending this argument, a given voting configuration can be generated by $\prod_{s=1}^t 2^{p_s-1}$ contexts. Each such bit string also occurs once for each context of $\mathcal{U}$. Thus, $2^{|\mathcal{U}|} \prod_{s=1}^t 2^{p_s-1}$ swaps of $V_n$ have this given voting configuration. Since all input parent sets are pairwise disjoint, this 
simplifies to $2^{n-t-1}$.
\end{proof}

The next four lemmas are stated, with proof details given in the appendix.

\begin{lemma}
Suppose each of the $2^t$ voting configurations occurs for the same number of swaps. Let $N$ be an optimal solution for $T$. Then 
\begin{equation*}
f_T(N)=\begin{cases}
          t \cdot 2^{n-2}-2^{n-t-1} \cdot c \binom{2c-1}{c} &\text{if } t=2c-1 \\
          t \cdot 2^{n-2}-2^{n-t-1} \cdot c \binom{2c}{c} &\text{if } t=2c \\
     \end{cases}
\end{equation*}
\end{lemma}

For a problem instance $T$ under our premises, this yields combinatorial expressions to compute the error made by an optimal CPT. Next, we give an expression for the error made by any of the input CPTs $N_s$, which we claim to be at most $4/3$ times the objective value of an optimal solution. 

\begin{lemma}
$f_T(N_s) = (t-1)\cdot2^{n-2}$ for all $s\in\{1,\ldots,t\}$.
\end{lemma}

\begin{lemma}
\label{lem:odd_t}
If $1\le s\le t=2c-1$, $c>1$ then $\frac{f_T(N_s)}{f_T(N)} \le \frac{4}{3}$.
\end{lemma}

\begin{lemma}
\label{lem:even_t}
If $1\le s\le t=2c$, $c>1$ then $\frac{f_T(N_s)}{f_T(N)} \le \frac{4}{3}$.
\end{lemma}

This finally completes the proof of Theorem~\ref{thm:symmetric}.

\section{Optimal Solution for a Fixed Parent Set}

Theorem~\ref{thm:2approximation} states that the \emph{best input CPT}\/ is never worse than the optimal solution by more than a factor of 2. This raises the question whether the \emph{best input parent set}\/ yields better guarantees on the approximation ratio. To this end, instead of taking the best input CPT unmodified as an aggregate output, we propose Algorithm~\ref{alg:approx}  as an approximation algorithm.

\begin{algorithm}[t]
    \SetKwInOut{Input}{Input}
    \SetKwInOut{Output}{Output}
    \Input{A tuple $T=(N_1,\ldots,N_t)$ of CPTs for $V_n$}
    \Output{$\CPT(N^a,V_n)$ with $Pa(N^a,V_n) \subseteq Pa(N_s,V_n)$ for some $s\in\{1,\ldots,t\}$}
    \caption{Build $\CPT(N^a,V_n)$ that minimizes $f_T$ subject to $Pa(N^a,V_n)\subseteq  Pa(N_s,V_n)$ for some $s\in\{1,\ldots,t\}$}
    \label{alg:approx}
    \For {each $s\in\{1,\ldots,t\}$}
    {
       compute a CPT $N^*_s$ that minimizes $f_T$ among all CPTs with parent set contained in $Pa(N_s,V_n)$
    }
    $N^a=N^*_{s^*}$ where $f_T(N^*_{s^*})=\min \{f_T(N^*_s)\mid 1\le s\le t\}$.
\end{algorithm}

Clearly, Algorithm~\ref{alg:approx} cannot produce worse outputs than the trivial algorithm that simply selects the best input CPT. What still needs to be addressed is (i) can this algorithm be designed to run in polynomial time?, and (ii) what approximation ratio can this algorithm achieve (under which circumstances)? To address (i), all that is needed is a polynomial-time algorithm that, given $T=(N_1,\ldots,N_t)$ and a parent set $P\in\{Pa(N_s,V_n)\mid 1\le s\le t\}$, produces a CPT $N^*_s$ that minimizes $f_T$ among all CPTs with parent set $Pa(N_s,V_n)$. With Algorithm~\ref{alg:NewOptimal}, we provide an algorithm that solves a more general problem: given $T=(N_1,\ldots,N_t)$ and \emph{any}\/ parent set $P\subseteq\bigcup_{1\le s\le n}Pa(N_s,V_n)$, it produces a CPT $N^*_s$ that minimizes $f_T$ among all CPTs with parent set contained in $P$. We will see below, that this algorithm runs in time polynomial in the size of its input, when the cardinality of $P$ is bounded by the size of the largest input parent set.

Given $T$ and $P$ as input, Algorithm \ref{alg:NewOptimal} computes, for each context $\gamma$ of $P$, the frequency of both possible preference orderings, and assigns the ordering $0\succ 1$ if it is in the majority, $1\succ 0$ otherwise. Lines 7-11 iterate over each $N_s$ in $T$, and count the number of rules in $\CPT(N_s,V_n)$ that would apply to a swap consistent with context $\gamma$ and  entail $0 \succ 1$. This count is multiplied by the number of swaps ordered by each rule in $CPT(N_s,V_n)$ and then divided by the number of possible contexts for $P \setminus Pa(N_s,V_n)$, since $P$ is fixed to $\gamma$. At each iteration,  \textit{zerovotes}\/ has the number of swaps for which $N_s$ entails $0 \succ 1$ given $\gamma$. In line 12, \textit{zerovotes}\/ has the total number of swaps for which $0 \succ 1$ is entailed given $\gamma$, summed over all input CPTs. The number of votes for $1 \succ 0$ is then found by subtracting this from the total number of votes given $\gamma$. Lines 13-17 then assign the ordering  $0\succ 1$ if it is in the majority, $1\succ 0$ otherwise. Repeating this for each possible context of $P$, we obtain $\CPT(N^a,V_n)$ with $2^{|P|}$ rules. Some of the attributes in $P$ might be \emph{irrelevant}\/ parents for this CPT in the sense that their value does not affect the preference order of the CPT, see \cite{koriche2010learning,allen2015cp}. Removing these attributes yields a more compact representation semantically equivalent to $N^a$. This can be done in time linear in the size of $N^a$ and quadratic in $n$  \cite{ali2021aggregating}.

\begin{algorithm}[t]
    \SetKwInOut{Input}{Input}
    \SetKwInOut{Output}{Output}
    \Input{$T=(N_1,\ldots,N_t)$, $P\subseteq\{V_1,\ldots,V_{n-1}\}$}
    \Output{An optimal $\CPT(N^a,V_n)$ wrt $f_T$ with $Pa(N^a,V_n) \subseteq P$}
    \caption{Build $\CPT(N^a,V_n)$ that minimizes $f_T$ with $Pa(N^a,V_n)\subseteq P$ for given $T$ and $P$}
    \label{alg:NewOptimal}
    \For {each $\gamma \in \Inst(P)$}
    {
        $zerovotes = 0$\\
        \For {each $s\in\{1,\ldots,t\}$}
        {
            $numRules$ = the number of rules in $\CPT(N_s,V_n)$ voting $0 \succ 1$ with contexts consistent with $\gamma$
            
            $numSwaps = numRules \cdot 2^{n-|Pa(N_s,V_n)|-1}$
            
            $numSwaps = numSwaps / 2^{|P-Pa(N_s,V_n)|}$
            
            $zerovotes = zerovotes+numSwaps$
        }
        $onevotes = t \cdot 2^{n-|P|-1} - zerovotes$\\
        \If {$zerovotes > onevotes$}
        {
           add $\gamma:0 \succ 1$ to $\CPT(N^a,V_n)$
        }
        \Else
        {
          add $\gamma:1 \succ 0$ to $\CPT(N^a,V_n)$
        }
      
      remove irrelevant parents from $\CPT(N^a,V_n)$
    }
\end{algorithm}

First, we show that Algorithm~\ref{alg:NewOptimal} is correct, i.e., it produces the optimal aggregate CPT with the given parent set.

\begin{theorem}
\label{thm:newalgooptimal}
Algorithm \ref{alg:NewOptimal} constructs $\CPT(N^a,V_n)$ such that $f_{T}(N^a) \le f_{T}(N)$ for all $N$ with $Pa(N,V_n)\subseteq P$.
\end{theorem}

\begin{proof}
Assume there is some $N$, $Pa(N,V_n)=P_N\subseteq P$, such that $f_T(N^a) > f_T(N)$. Then there exists $\gamma \in \Inst(P_N)$ for which $N$ and $N^a$ entail different orders, and $N$ disagrees with $T$ on fewer swaps consistent with $\gamma$ than $N^a$ does. Now $\gamma$ corresponds to some sub-matrix $M'$. The corresponding sub-matrices for $N$ and $N^a$ contain a constant value each---one of them 0, the other 1. Wlog, assume $N^a$ has an all-ones sub-matrix in place of $M'$. By Algorithm~\ref{alg:NewOptimal}, this implies $freq_{M'}(0 \succ 1) \le freq_{M'}(1 \succ 0)$. By our assumption, $N'$ now has the all-zeros sub-matrix in place of $M'$, and this all-zeros matrix has fewer inconsistencies with $M'$ than $N^a$'s all-ones matrix. From this, we have $freq_{M'}(0 \succ 1) > freq_{M'}(1 \succ 0)$---a contradiction.
\end{proof}

Second, Algorithm \ref{alg:NewOptimal} runs in polynomial time, when $|P|$ is bounded by the size of the largest input parent set.

\begin{theorem}
\label{thm:complexity} Algorithm \ref{alg:NewOptimal} runs in time $O(2^{|P|} \cdot \sum_{N_s \in T} |\CPT(N_s,V_n)|)$. 
In particular, if $|P|\le\max\{|Pa(N_s,V_n)\mid 1\le i\le t\}$, it runs in time polynomial in $\sum_{1\le s\le t} |\CPT(N_s,V_n)|$. Moreover, Algorithm~\ref{alg:approx} runs in polynomial time.
\end{theorem}

\begin{proof}
The total runtime of the inner loop (lines 7-11 iterated) is in $O(\sum_{1\le s\le t} |CPT(N_s,V_n)|)$, which is linear in input size. The inner loop runs once per $\gamma \in\Inst(P)$, i.e., $2^{|P|}$ times. Since removing irrelevant parents can be done in time linear in the size of $N^a$ and quadratic in $n$, this yields a runtime in $O(2^{|P|} \cdot \sum_{1\le s\le t} |CPT(N_s,V_n)| )$. The remaining statements of the theorem follow immediately.
\end{proof}

Thus Algorithm~\ref{alg:approx} is an efficient method providing the optimal CPT wrt any of the input parent sets. It clearly cannot produce an output worse than the trivial algorithm, which simply outputs the best input CPT. To demonstrate that it can do substantially better, we define a family of input instances that provides useful insights into CPT aggregation.

\begin{definition}
For $n\ge 3$ and $k\in\{2,\ldots,n-1\}$ define $T^{k,n}=(N^{k,n}_1,\ldots,N^{k,n}_t)$ as follows. Let $t=\binom{n-1}{k}2^k$. 
Then each number $s\in\{1,\ldots,t\}$ corresponds to a unique pair $(P,\gamma)$, where $P$ is a $k$-element subset of $\mathcal{V}\setminus{V_n}$ and $\gamma$ is a context over $P$. Now let $N^{k,n}_s$ be the CPT with rules $\gamma:1\succ 0$, and $\gamma':0\succ 1$ for all contexts $\gamma'\in\Inst(P)\setminus\{\gamma\}$.
\end{definition}

For the case $k=n-1$, $(T^{k,n})_{k,n}$ was already mentioned by \citet{endriss2014binary}. 
It will turn out (Theorem~\ref{thm:optimalApproximation}) that Algorithm~\ref{alg:approx} produces an optimal solution for $T^{k,n}$ when $n\ge 3$ and $2\le k\le n-1$. By contrast, for $k=n-1$, \citet{endriss2014binary} proved that the trivial algorithm, which outputs the best input CPT, cannot obtain an approximation ratio better than 2 for the family $(T^{n-1,n})_{n\ge 3}$  (Theorem~\ref{thm:2-epsilon}). Moreover, we will argue that the trivial algorithm provides solutions whose objective value is at least a factor of $3/2$ above the optimum (Theorem~\ref{thm:k=2}).

\begin{theorem}\label{thm:optimalApproximation}
    Let $n\ge 3$ and $k\in\{2,\ldots,n-1\}$. Then Algorithm~\ref{alg:approx} outputs an optimal solution for $T^{k,n}$.
\end{theorem}

\begin{proof}
    Each $k$-element subset $\mathcal{V'}\subseteq \mathcal{V}\setminus{V_n}$ is the parent set of $2^k$ input CPTs. By definition of $T^{k,n}$, for any context over any $\mathcal{V'}$, all but one of the CPTs with  $\mathcal{V'}$ as parent set entail $0 \succ 1$. Thus, for all swaps of $V_n$, $0 \succ 1$ is the majority ordering and the optimal solution is the separable $0 \succ 1$. Now assume $Pa(N^{k,n}_s,V_n)$ is the input to Algorithm \ref{alg:NewOptimal}, for some $s\in\{1,\ldots,\binom{n-1}{k}2^k$. Consider the $2^k$ sub-matrices corresponding to contexts of $Pa(N^{k,n}_s,V_n)$, each with $2^{n-k-1}$ rows. On each row, $0 \succ 1$ occurs more often than $1 \succ 0$. So, over each sub-matrix $M$, $freq_M(0 \succ 1)>freq_M(1 \succ 0)$. Thus Algorithm \ref{alg:NewOptimal} outputs the optimal separable $0 \succ 1$, and Algorithm \ref{alg:approx} outputs an optimal solution.
\end{proof}

\begin{theorem}[cf. \cite{endriss2014binary}]\label{thm:2-epsilon}
    Let $\varepsilon >0$ be any positive real number. Then there is some $n\ge 3$ such that 
    $$f_{T^{n-1,n}}(N^{n-1,n}_s)> (2-\varepsilon) f_{T^{n-1,n}}(N)$$
    for all $1\le s\le2^{n-1}$ ($=\binom{n-1}{n-1}2^{n-1}$), where $N$ is any optimal solution for $T^{n-1,n}$.
\end{theorem}

\begin{theorem}\label{thm:k=2}
    Let $n\ge 3$, $k\in\{2,\ldots,n-1\}$, and $1\le s\le \binom{n-1}{k}2^k$. Let $N$ be an optimal solution for $T^{k,n}$. Then $f_{T^{k,n}}(N^{k,n}_s)\ge (3/2) f_{T^{k,n}}(N)$.
\end{theorem}

In order to prove this theorem, we will need to establish two helpful lemmas.

\begin{lemma}\label{lem:opt2}
    Let $N$ be an optimal solution for $T^{k,n}$. Then $f_{T^{k,n}}(N) = 2^{n-1} \binom{n-1}{k}$.
\end{lemma}

\begin{proof}
    Each $N^{k,n}_2$ in $T^{k,n}$ entails $1 \succ 0$ for one rule applying to $2^{n-k-1}$ swaps. By the proof of Theorem \ref{thm:optimalApproximation}, $N$ is the separable $0 \succ 1$, implying $\Delta(N,N^{k,n}_s)=2^{n-k-1}$. Summing up over all $t$ values of $s$ proves the claim.
\end{proof}

\begin{lemma}\label{lem:apr}
    Let $1\le s\le \binom{n-1}{k}2^k$. Then $f_{T^{k,n}}(N^{k,n}_s)= (2^n - 2^{n-k}) \sum_{k'=0}^{k} \binom{k}{k'} \binom{n-k-1}{k-k'}$.
\end{lemma}

\begin{proof}
Note that $f_{T^{k,n}}(N^{k,n}_s)=\sum_{s'\ne s}\Delta(N^{k,n}_s,N^{k,n}_{s'})$. The value $\Delta(N^{k,n}_s,N^{k,n}_{s'})$ depends on $|Pa(N^{k,n}_s,V_n) \cap Pa(N^{k,n}_{s'},V_n)|=:k'$, as well as the contexts $\gamma_1$ and $\gamma_2$ for which $N^{k,n}_s$ and $N^{k,n}_{s'}$, resp., entail $1 \succ 0$. Let $P$ denote $Pa(N^{k,n}_{s'},V_n)$ for \emph{some}\/ $s'\ne s$, $1\le s'\le \binom{n-1}{k}2^k$.

Given $0 \le k' \le k$, there are $\binom{k}{k'} \binom{n-k-1}{k-k'}$ parent sets $P$ of size $k$ such that $Pa(N^{k,n}_s,V_n) \cap P =k'$. For each such $P$, if $k'=k$, the tuple $T$ has $2^{k}-1$ CPTs other than $N^{k,n}_s$ that have parent set $P$. For $k'<k$, $T$ has $2^k$ CPTs with parent set $P$.  
While it is not necessary to treat the cases $k'=k$ and $k'=0$ separately, we still do so, as it may help the reader better understand our argument for general values of $k'$.

    For $k'=k$ there are $2^k$ contexts of $Pa(N^{k,n}_s,V_n) \cup P$ ($=Pa(N^{k,n}_s,V_n)$). By construction, $N^{k,n}_s$ and $N^{k,n}_{s'}$ disagree on $2$ contexts and thus on $2^{n-k}$ swaps, i.e., $\Delta(N^{k,n}_s,N^{k,n}_{s'})=2^{n-k}$ for each of the $2^{k}-1$ CPTs $N^{k,n}_{s'}$ other than $N^{k,n}_s$ that have parent set $P$. 

    For $k'=0$ there are $2^{2k}$ contexts of $Pa(N^{k,n}_s,V_n) \cup P$. 
    $N^{k,n}_s$ entails $0 \succ 1$ for $2^k-1$ contexts and $N^{k,n}_{s'}$ entails $1 \succ 0$ for one context. These contexts are independent of each other since $k'=0$. This yields $(2^k-1)\cdot 1= 2^k-1$ contexts of $Pa(N^{k,n}_s,V_n) \cup P$ on which the two CPTs disagree. Accounting also for the symmetric case with the roles of $N^{k,n}_s$ and $N^{k,n}_{s'}$ exchanged, we obtain a disagreement on $2^{k+1}-2$  contexts of $Pa(N^{k,n}_s,V_n) \cup P$, each corresponding to $2^{n-2k-1}$ swaps. This yields $\Delta(N^{k,n}_s,N^{k,n}_{s'})=2^{n-k}-2^{n-2k}$ for each  CPT $N^{k,n}_{s'}$ that has a parent set $P$ disjoint from $Pa(N^{k,n}_s,V_n)$, i.e., a parent set $P$ yielding $k'=0$. There are $2^k\binom{n-k-1}{k}$ such CPTs, namely $2^k$ CPTs for each choice of $k$-element set $P$ disjoint from $Pa(N^{k,n}_s,V_n)$.
    
    Lastly, for $1 \le k' \le k-1$, there are $2^{2k-k'}$ contexts of $Pa(N^{k,n}_s,V_n) \cup P$. Note that, for $\gamma_1\in\Inst(Pa(N^{k,n}_s,V_n))$ and $\gamma_2\in\Inst(P)$, we have: ($N^{k,n}_s$ has $\gamma_1:1\succ 0$ \emph{and}\/ $N^{k,n}_{s'}$ has $\gamma_2:1\succ 0$) iff $\gamma_1$ and $\gamma_2$ are consistent, i.e., they have the same values on all attributes in $Pa(N^{k,n}_s,V_n) \cap P$. Now suppose $N^{k,n}_s$ entails $1 \succ 0$ for $\gamma_1$. There are $2^{k-k'}$ CPTs with parent set $P$ that entail $1 \succ 0$ for some $\gamma_2$ consistent with $\gamma_1$, and $2^k-2^{k-k'}$ CPTs with parent set $P$ that do not.

Consider CPTs of the first type. There are $2^{k}$ contexts of $Pa(N^{k,n}_s,V_n)$; for exactly one such context, namely $\gamma_1$, the CPT $N^{k,n}_s$ entails $1 \succ 0$. Each of the $2^{k-k'}$ contexts of $P \setminus Pa(N^{k,n}_s,V_n)$ can be appended to $\gamma_1$. For $2^{k-k'}-1$ of these, $N^{k,n}_{s'}$ entails $0 \succ 1$; so $N^{k,n}_s$ and $N^{k,n}_{s'}$ disagree on all such contexts appended to $\gamma_1$. Also counting the symmetric case where $N^{k,n}_{s'}$ entails $1 \succ 0$, $N^{k,n}_s$ and $N^{k,n}_{s'}$ disagree on $2^{k-k'+1}-2$ contexts of $Pa(N^{k,n}_s,V_n) \cup P$, each of which orders $2^{n-2k+k'-1}$ swaps. Thus $\Delta(N^{k,n}_s,N^{k,n}_{s'})=2^{n-k}-2^{n-2k+k'}$ in this case.

    Consider CPTs of the second type. There are $2^{k}$ contexts of $Pa(N^{k,n}_s,V_n)$; for exactly one such context, namely $\gamma_1$, the CPT $N^{k,n}_s$ entails $1 \succ 0$.  Each of the $2^{k-k'}$ contexts of $P \setminus Pa(N^{k,n}_s,V_n)$ can be appended to $\gamma_1$. Since CPTs of the second type entail $0 \succ 1$ for every $\gamma_2\in\Inst(P)$ that is consistent with $\gamma_1$, we have that $N^{k,n}_{s'}$ entails $0 \succ 1$ for all $2^{k-k'}$ completions of $\gamma_1$ in $\Inst(P)$. Counting also the symmetric case where $N^{k,n}_{s'}$ entails $1 \succ 0$, $N^{k,n}_s$ and $N^{k,n}_{s'}$ disagree on $2^{k-k'+1}$ contexts of $Pa(N^{k,n}_s,V_n) \cup P$, each of which orders $2^{n-2k+k'-1}$ swaps.  Thus $\Delta(N^{k,n}_s,N^{k,n}_{s'}=2^{n-k}$ in this case.

    Combining these pieces, the value of $f_{T^{k,n}}(N^{k,n}_s)$ equals
    \begin{align*}
    && (2^k-1)\cdot2^{n-k}+2^k \binom{n-k-1}{k}(2^{n-k}-2^{n-2k}) \\
    &+ & \sum_{k'=1}^{k-1}\big[2^{k-k'} \binom{k}{k'}\binom{n-k-1}{k'}(2^{n-k}-2^{n-2k+k'})\\
    &&\ \ +2^k-2^{k-k'}\binom{k}{k'}\binom{n-k-1}{k'}2^{n-k}\big]
    \end{align*}



Simplifying this with straightforward calculations yields
    $f_{T^{k,n}}(N^{k,n}_s)= (2^{n}-2^{n-k}) \sum_{k'=0}^{k} \binom{k}{k'}\binom{n-k-1}{k'}$.
\end{proof}

\noindent\emph{Proof of Theorem~\ref{thm:k=2}.}\/ Follows from Lemmas \ref{lem:opt2} and \ref{lem:apr}, using Vandermonde's identit; see appendix for details.
\qed

\section{Conclusions}

Since CP-net aggregation (wrt swap preferences) is known to be intractable, the design and analysis of approximation algorithms for preference aggregation is one of few viable approaches to efficient preference aggregation in this context. Proposition~\ref{prop:fixedt} implies that optimal CP-net aggregation is intractable not due to any difficulties in scaling with the number of attributes, but just due to the difficulty of scaling with the number of input CPTs. In particular, the cause of intractability lies solely in the parent set size of optimal solutions, which can be  asymptotically larger than the size of the largest input parent set. 

Therefore, we focused on approximation algorithms that keep the size of the output parent set linear in the size of the largest input parent set. A trivial such algorithm is one that simply outputs the best input CPT, which yields a 2-approximation in general. When imposing a symmetry constraint on the input CPT, the approximation ratio of this algorithm improves from 2 to $4/3$, but in general, the ratio can be arbitrarily close to $2$ (see Theorem~\ref{thm:2-epsilon}). 

Algorithm~\ref{alg:approx} instead considers each input \emph{parent set}\/ and calculates a \emph{provably}\/ optimal CPT for that parent set. Finally it outputs the best thus attained CPT. This polynomial-time method is never worse than the trivial algorithm, yet substantially better for some families of input instances. At the time of writing this paper, we are not aware of any problem instance on which Algorithm~\ref{alg:approx} attains an approximation ratio greater than $4/3$. One open problem is to either prove that $4/3$ is indeed an upper bound on this algorithm's approximation ratio, or else to provide a problem instance for which Algorithm~\ref{alg:approx} has a ratio exceeding $4/3$.

Due to the relations between binary aggregation and CP-net aggregation, we hope that our work provides insights that are useful beyond the aggregation of CP-nets. 


\paragraph{Acknowledgements.} The authors would like to thank Kamyar Khodamoradi for helpful discussions. Boting Yang was supported in part by an NSERC Discovery Research Grant, Application No.: RGPIN-2018-06800. Sandra Zilles was supported in part by an NSERC Discovery Research Grant, Application No.: RGPIN-2017-05336, by an NSERC Canada Research Chair, and by a Canada CIFAR AI Chair held at the Alberta Machine Intelligence Institute (Amii).

\bibliography{aaai24}

\newpage

\section{Appendix}

\section{An example CP-net}

\begin{figure}[H]
    \centering
    \includegraphics[scale=0.375]{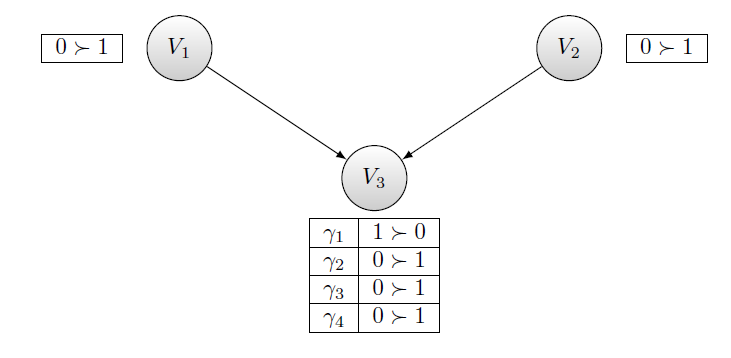}
    \caption{$N_1$}
    \label{fig:N1}
\end{figure}

In Figure \ref{fig:N1}, $N_1$ is defined over $V=\{V_1,V_2,V_3\}$. For attributes $V_1$ and $V_2$, the individual with $N_1$ as their preference model unconditionally prefers $0 \succ 1$. For $V_3$, their preference is conditioned on the values assigned to $V_1$ and $V_2$. Each vertex, denoting an attribute, is annotated with a Conditional Preference Table (CPT), with a total preference ordering given for each possible instantiation of the corresponding parent set. Any outcome pair that differs on the value assigned to exactly one attribute is called a swap, e.g. $000,001$ is a swap of $V_3$. Assuming $\gamma_1=00$, the best outcome for $N_1$ is $001$, and the worst outcome is $000$. The outcome pair $001,101$ cannot be ordered by $N_1$, because CP-net semantics does not tell us how to trade-off between the best value for $V_1$ and that for $V_2$.

\section{Matrix Representation}

Notice that $N_1$ in Figure \ref{fig:N1} entails $1 \succ 0$ for only one context of $Pa(N_1,V_3)$, namely $\gamma_1=00$. We now define a tuple $T$ of four CP-nets. Each $N_i \in T$ has identical dependency graph, and identical CPTs for $V_1$ and $V_2$. The only difference is that for each $N_i$, $CPT(N_i,V_3)$ entails $1 \succ 0$ for $\gamma_i$, and $0 \succ 1$ for all other contexts. We now give the matrix representation for $T$. Each row represents an instantiation of $\{V_1,V_2\}$, corresponding to a different swap of $V_3$. Each column represents $N_i \in T$.

\begin{tabular}{|c|c|c|c|c|}
    \hline
    $o[\{V_1,V_2\}]$ & $N_1$ & $N_2$ & $N_3$ & $N_4$ \\
    \hline
     $00$ & $1$ & $0$ & $0$ & $0$ \\
     \hline
     $01$ & $0$ & $1$ & $0$ & $0$ \\
     \hline
     $10$ & $0$ & $0$ & $1$ & $0$ \\
     \hline
     $11$ & $0$ & $0$ & $0$ & $1$ \\
     \hline
\end{tabular}

In fact, $T$ is an instance of the special family of input instances mentioned in Definition 15. In particular, $T$ is $T^{2,3}$. It is easy to see that $f_T(N)=4$ for an optimal solution $N$. We can also verify that none of the inputs attain a ratio strictly less than $\frac{3}{2}$, while Algorithm 1 produces an optimal solution.

\section{Proofs of Lemmas For Theorem~6}

\setcounter{theorem}{6}

Lemmas \ref{lem:opt} and \ref{lem:SameCount} were proven in the main body of the paper:

\begin{lemma}
\label{lem:opt}
Let $T'=(N'_1,\ldots,N'_{t})$ be any problem instance of $t$ CPTs. Assuming each of the $2^t$ voting configurations occurs exactly once in $M(T')$, and $N$ is an optimal solution for $T'$, we have 
\begin{equation*}
f_{T'}(N)=\begin{cases}
          2 \cdot \sum_{\kappa = 0}^{c-1} \kappa \binom{2c-1}{\kappa} &\text{if }t=2c-1 \\
          2 \cdot \sum_{\kappa = 0}^{c-1} \kappa \binom{2c}{\kappa} + c \binom{2c}{c}  &\text{if }t=2c \\
     \end{cases}
\end{equation*}
\end{lemma}

\begin{proof}
    See main body of the paper.
\end{proof}

\begin{lemma}
\label{lem:SameCount}
Each of the $2^t$ voting configurations of $T$ is the row vector for exactly $2^{n-t-1}$ swaps.
\end{lemma}
\begin{proof}
    See main body of the paper.
\end{proof}

These lemmas are used in the proofs of the following four lemmas.

\begin{lemma}
Suppose each of the $2^t$ voting configurations occurs for the same number of swaps. Let $N$ be an optimal solution for $T$. Then 
\begin{equation*}
f_T(N)=\begin{cases}
          t \cdot 2^{n-2}-2^{n-t-1} \cdot c \binom{2c-1}{c} &\text{if } t=2c-1 \\
          t \cdot 2^{n-2}-2^{n-t-1} \cdot c \binom{2c}{c} &\text{if } t=2c \\
     \end{cases}
\end{equation*}
\end{lemma}

\begin{proof}
Lemma \ref{lem:opt} provides the total error of $N$ when each voting configuration occurs for exactly one swap. If each voting configuration occurs for the same number of swaps, we can simply multiply the error made on each, by the number of swaps for which that configuration occurs. For the tuples satisfying the conditions of Lemma \ref{lem:SameCount}, we obtain $f_T(N)=$
\begin{equation*}
\begin{cases}
          2^{n-t}\cdot \sum_{\kappa=0}^{c-1} \kappa\binom{2c-1}{\kappa} &\text{if} \, t=2c-1 \\
          2^{n-t} \cdot \sum_{\kappa=0}^{c-1} \kappa\binom{2c}{\kappa} + 2^{n-t-1} \cdot c \binom{2c}{c}  &\text{if} \, t=2c \\
     \end{cases}
\end{equation*}
which simplifies to the desired expression.
\end{proof}

For a problem instance $T$ under our premises, this gives us combinatorial expressions to compute the error made by an optimal CPT. Next, we give an expression to compute the error made by any of the input CPTs $N_s$, which we claim to be at most $4/3$ times the objective value of an optimal solution. 

\begin{lemma}
$f_T(N_s) = (t-1)\cdot2^{n-2}$ for all $s\in\{1,\ldots,t\}$.
\end{lemma}

\begin{proof}
Consider any $s'\ne s$. Since the CPTs in the tuple $T$ are symmetric and have pairwise disjoint parent sets, $N_s$ and $N_{s'}$ disagree on half of all swaps, i.e., on $2^{n-2}$ swaps. Thus $f_T(N_s) = (t-1)\cdot 2^{n-2}$.
\end{proof}


\begin{lemma}
\label{lem:odd_t}
If $1\le s\le t=2c-1$, $c>1$ then $\frac{f_T(N_s)}{f_T(N)} \le \frac{4}{3}$.
\end{lemma}

\begin{proof}
This is equivalent to proving
$$ (2c-1)\cdot2^{n-2} - 2^{n-2c} \cdot c \binom{2c-1}{c} \ge \frac{3}{4} (2c-2)\cdot2^{n-2} $$

$$\Leftarrow (2c-1)\cdot2^{n-2} - 2^{n-2c} \cdot c \binom{2c-1}{c} \ge (6c-6)\cdot2^{n-4} $$

$$\Leftarrow (8c-4)\cdot2^{n-4} - 2^{n-2c} \cdot c \binom{2c-1}{c} \ge (6c-6)\cdot2^{n-4} $$

$$\Leftarrow 2^{n-2c} \cdot c \binom{2c-1}{c} \le (2c+2)\cdot2^{n-4} $$

    $$ \Leftarrow c \binom{2c-1}{c} \le (2c+2) \cdot 2^{2c-4} $$

$$ \Leftarrow c \binom{2c-1}{c} \le (c+1) \cdot 2^{2c-3} $$

We prove the latter for all $c > 1$ using induction.

When $c=2$, the inequality obviously holds.

Assume the inequality holds for some fixed $c$
and we need to prove that

$$ (c+1)\binom{2c+1}{c+1} \le (c+2) \cdot 2^{2c-1}$$

We know $$\binom{2c+1}{c+1} = \binom{2c+1}{c} = \frac{2c+1}{c+1} \binom{2c}{c} = \frac{4c+2}{c+1} \binom{2c-1}{c}$$

and $$ (c+1)\binom{2c+1}{c+1} = (4c+2) \binom{2c-1}{c} $$

By inductive hypothesis,
$$ c \binom{2c-1}{c} \le (c+1) \cdot 2^{2c-3}\,, $$
so, multiplying both sides by $4$
$$ 4c \binom{2c-1}{c} \le (c+1) \cdot 2^{2c-1} \,.$$
Adding $2 \cdot \binom{2c-1}{c}$ on both sides yields
$$ 4c \binom{2c-1}{c} + 2 \cdot \binom{2c-1}{c} \le (c+1) \cdot 2^{2c-1} + 2 \cdot \binom{2c-1}{c} $$

$$ (4c+2) \binom{2c-1}{c} \le (c+1) \cdot 2^{2c-1} + \binom{2c}{c} $$

$$ (c+1) \binom{2c+1}{c+1} \le (c+1) \cdot 2^{2c-1} + \binom{2c}{c} $$

We complete the induction by proving
$$ \binom{2d}{d} \le 2^{2d-1} $$
for all $d\ge 1$, also using induction. 

For $d=1$, this is obviously true.

Assume $ \binom{2d}{d} \le 2^{2d-1} $ for a fixed $d$.
We have to prove $\binom{2d+2}{d+1} \le 2^{2d+1} $.
We know $\binom{2d+2}{d+1} = \frac{2 \cdot (2d+1)}{d+1} \binom{2d}{d}$.
By inductive hypothesis $ \binom{2d}{d} \le 2^{2d-1} $. Multiplying both sides by $\frac{2 \cdot (2d+1)}{d+1}$ gives us

$$ \frac{2 \cdot (2d+1)}{d+1} \binom{2d}{d} \le \frac{2 \cdot (2d+1)}{d+1} \cdot  2^{2d-1} $$

$$ \binom{2d+2}{d+1} \le \frac{2 \cdot (2d+1)}{d+1} \cdot  2^{2d-1} \le 4 \cdot 2^{2d-1} $$ 

The above holds because $\frac{2 \cdot (2d+1)}{d+1} < 4$. This proves $ \binom{2d}{d} \le 2^{2d-1} $ for all values of $d\ge 1$, which in turn completes our first induction.
\end{proof}

\begin{lemma}
\label{lem:even_t}
If $1\le s\le t=2c$, $c>1$ then $\frac{f_T(N_s)}{f_T(N)} \le \frac{4}{3}$.
\end{lemma}

\begin{proof}
This is equivalent to proving

$$2c\cdot2^{n-2} - 2^{n-2c-1} \cdot c \binom{2c}{c} \ge \frac{3}{4} (2c-2)\cdot2^{n-2} $$

$$ \Leftarrow 2c\cdot2^{n-2} - 2^{n-2c-1} \cdot c \binom{2c}{c} \ge (6c-3)\cdot2^{n-4} $$

$$ \Leftarrow 8c\cdot2^{n-4} - 2^{n-2c-1} \cdot c \binom{2c}{c} \ge (6c-3)\cdot2^{n-4} $$

$$ \Leftarrow 2^{n-2c-1} \cdot c \binom{2c}{c} \le (2c+3)\cdot2^{n-4} $$

$$ \Leftarrow c \binom{2c}{c} \le (2c+3)\cdot2^{2c-3} $$

$$ \Leftarrow c \binom{2c-1}{c} \le (2c+3)\cdot2^{2c-4} $$

In the proof of Lemma \ref{lem:odd_t}, we established $c \binom{2c-1}{c} \le (c+1)\cdot2^{2c-3} =  (2c+2)\cdot2^{2c-4}$. Since $c>1$, the above immediately follows.

\end{proof}

\section{Proof of Theorem~\ref{thm:k=2}}

\setcounter{theorem}{18}
Lemmas \ref{lem:opt2} and \ref{lem:apr} were proven in the main body of the paper:

\begin{lemma}\label{lem:opt2}
    Let $N$ be an optimal solution for $T^{k,n}$. Then $f_{T^{k,n}}(N) = 2^{n-1} \binom{n-1}{k}$.
\end{lemma}

\begin{lemma}\label{lem:apr}
    Let $1\le s\le \binom{n-1}{k}2^k$. Then $f_{T^{k,n}}(N^{k,n}_s)= (2^n - 2^{n-k}) \sum_{k'=0}^{k} \binom{k}{k'} \binom{n-k-1}{k-k'}$.
\end{lemma}

\setcounter{theorem}{17}
\begin{theorem}\label{thm:k=2}
    Let $n\ge 3$, $k\in\{2,\ldots,n-1\}$, and $1\le s\le \binom{n-1}{k}2^k$. Let $N$ be an optimal solution for $T^{k,n}$. Then $f_{T^{k,n}}(N^{k,n}_s)\ge (3/2) f_{T^{k,n}}(N)$.
\end{theorem}

\begin{proof}
    By Lemmas \ref{lem:opt2} and \ref{lem:apr}, and since $(3/2)2^{n-1}=3\cdot2^{n-2}$, we need to show
    $$ (2^n - 2^{n-k}) \sum_{k'=0}^{k} \binom{k}{k'} \binom{n-k-1}{k-k'} \ge 3\cdot 2^{n-2} \binom{n-1}{k}\,.$$
    Since $k \ge 2$, we have $2^n-2^{n-k}\ge 2^n-2^{n-2}=3\cdot 2^{n-2}$. Thus it suffices to show
%
%
%
%
    $$\sum_{k'=0}^{k} \binom{k}{k'} \binom{n-k-1}{k-k'} \ge \binom{n-1}{k}\,,$$
which holds with equality by Vandermonde's Identity.
\end{proof}

\end{document}